\documentclass[12pt,letterpaper]{article}
\usepackage[T1]{fontenc}
\usepackage{lmodern}
\usepackage[utf8]{inputenc}
\usepackage{geometry}
\usepackage{booktabs}
\usepackage{amsmath,amssymb,amsfonts,amsthm}
\usepackage{cancel}
\usepackage{bm}
\usepackage{mathtools}
\usepackage{changepage}
\usepackage{verbatim}
\usepackage{graphicx}
\usepackage{emptypage}
\usepackage{newlfont}
\usepackage[all,cmtip]{xy}
\usepackage[bottom]{footmisc}
\usepackage[titletoc,title]{appendix}
\usepackage{chngcntr}
\usepackage{apptools}
\usepackage{listings}
\usepackage{color}
\usepackage{sgame, tikz} 
\usepackage{pgfplots}
\usepackage{authblk}
  
\usepackage{natbib}
\AtAppendix{\counterwithin{lem}{section}}
\usepackage{caption}
\usepackage{subcaption}
\usepackage{float}
\usepackage{xr-hyper}

\usepackage{sgamevar}
\usepackage{lipsum}
\usepackage{array}
\usepackage{enumitem}

\usepackage{kpfonts}

\usepackage{setspace}
\reversemarginpar

\newcommand{\argmax}{\operatornamewithlimits{argmax}}

\theoremstyle{definition}

\newtheorem*{defn*}{Definition}
\newtheorem{proposition}{Proposition}
\newtheorem*{prop*}{Proposition}

\newtheorem{ex}{Example}
\newtheorem*{ex*}{Example}
\newtheorem{lemma}{Lemma}

\newtheorem{assumption}{Assumption}
\theoremstyle{remark}
\newtheorem*{remarking}{Remark}

\theoremstyle{plain}
\newtheorem{theorem}{Theorem}

\newtheorem{claim}{Claim}

\renewcommand{\Pr}{\mathbb{P}}

\newcommand{\Rev}{\pi}

\newcommand{\spa}{second-price auction }
\newcommand{\fpa}{first-price auction }
\newcommand{\spadot}{second-price auction}
\newcommand{\fpadot}{first-price auction}
\DeclareMathOperator*{\Rex}{\overline{\mathbb{R}}}

\PassOptionsToPackage{hyphens}{url}\usepackage[colorlinks=true,linkcolor=blue, allcolors=blue]{hyperref}
\usepackage[nameinlink]{cleveref}
\crefname{assumption}{assumption}{assumptions}
\crefname{innercustomthm}{theorem}{theorems}
\crefname{proposition}{proposition}{propositions}
\crefname{ex}{example}{examples}
\crefname{definition}{definition}{definitions}
\crefname{claim}{claim}{claims}
\crefname{lemma}{lemma}{lemmas}
\crefname{theorem}{theorem}{theorems}
\crefname{remarking}{remark}{remarks}

\pgfplotsset{compat=1.16}

\title{Ads in Conversations}

\author[1,2]{Martino Banchio}
\author[2]{Aranyak Mehta}
\author[2]{Andres Perlroth}

\affil[1]{Università Bocconi, IGIER}
\affil[2]{Google Research}

\date{}

\begin{document}

\maketitle

\begin{abstract}
    We study the optimal placement of advertisements for interactive platforms like conversational AI assistants. 
    Importantly, conversations add a feature absent in canonical search markets --- time. The evolution of a conversation is informative about ad qualities, thus a platform could delay ad delivery to improve selection.
    However, delay endogenously shapes the supply of quality ads, possibly affecting revenue. We characterize the equilibria of first- and second-price auctions where the platform can commit to the auction format but not to its timing. We document
    sharp differences in the mechanisms' outcomes: first-price auctions are efficient but delay ad delivery, while second-price auctions avoid delay but allocate inefficiently. Revenue may be arbitrarily larger in a second-price auction than in a first-price auction. Optimal reserve prices alleviate these differences but flip the revenue ordering. 
\end{abstract}

\section{Introduction}

Conversational AI platforms, like chatbots and virtual assistants, are rapidly gaining popularity. These platforms offer users a convenient and interactive way to access information, complete tasks, and engage in entertainment. From booking flights and ordering food to playing games and getting personalized recommendations, conversational AI platforms have the potential to transform the way people interact with technology. Many of these platforms are currently offered free of charge in their basic version, but as their user base grows and their capabilities expand, it is natural to envision that, among a number of other strategies,\footnote{There are many alternative avenues for monetization of Conversational AI assistants, e.g., subscriptions, etc. We leave to future work the interesting question of whether advertising is the optimal monetization method in such a context, or instead commission-based fee structures could be more appropriate.} advertising could play a significant role in their monetization. Perplexity.ai, one of the leading conversation-based search bots, has recently begun offering advertising material during the conversation with its product.\footnote{ \url{https://www.perplexity.ai/hub/blog/why-we-re-experimenting-with-advertising}} The launch article mentions that ``\textit{advertising is the best way to ensure a steady and scalable revenue stream.}'' The sentiment appears to be shared by publishers and advertisers alike.\footnote{\url{https://www.adweek.com/media/why-marketers-welcome-ads-in-chatbots/}}

This presents a unique opportunity for targeted advertising. Unlike traditional online advertising, conversational AI platforms can glean real-time insights into user preferences through the natural flow of conversation. For example, a user interacting with a chatbot to book a hotel might reveal their budget, desired location, travel dates, and other preferences through their questions and responses. This dynamic information can be used to refine estimates of ad quality, measured in this paper by the probability of a user clicking on an ad (click-through rate).\footnote{Of course, click-through rates are only one of many possible measures of ad quality. The model will apply to generic ad quality measures as long as advertisers' payments depend on their quality.}

We study how to leverage this dynamic information in advertising auctions. A platform commits to an auction format but she chooses when to run the auction and display one ad only after receiving the advertisers' bids. Her decision is then best understood as a real options problem. She can delay the ad selection to acquire more precise quality estimates, increasing expected payments from high quality advertisers. At the same time, identifying low quality advertisers intrinsically limits the set of relevant competitors. Delay reduces market thickness, decreases competition and ultimately hurts revenue.

We develop a theoretical model to characterize this tension between information acquisition and market thickness. 
In the base model a platform (the auctioneer, she) wants to sell a single ad to one of two advertisers (buyer, he) over the course of a conversation. Each advertiser's value of showing the ad is the product of two components.
First, he values a click at $v_i$, and a user clicks on the ad with probability $\theta_i$. We model $\theta_i$ as a binary variable, where $\theta_i = 1$ indicates a good match (high click-through rate) and $\theta_i =0 $ indicates a bad match (low click-through rate). Initially, the quality score is unknown to everyone, including the advertisers.
As the user interacts with the conversational AI, the platform gradually learns whether each ad is a good or a bad match.
We allow the platform full flexibility for her information acquisition process.

We assume that the platform can commit to the auction format but not to its timing. This partial commitment assumption reflects the limited monitoring available to advertisers in online advertising settings.\footnote{
While it may be possible to ex-post verify the nature of the auction format, verifying that the platform kept its timing promises seems less plausible.}
When the platform decides to run its auction, bids are scored (i.e. ranked by the products of bid and click-through rate), the winner's ad is shown and he pays only if the ad is clicked by the user. Formally, we study the equilibria of two of the most common auction formats, first- and second-price auction, where the auction timing is optimally chosen by the platform.  We thus begin by characterizing the optimal exercise of the platform's \emph{option}: that is, the auction timing chosen by the auctioneer. We show that in equilibrium the platform delays a \fpa as much as possible, and she expedites a \spa instead.

First-price auctions realize revenue by selecting the efficient advertisement. The platform delays the allocation until it has precise information about the quality scores, ensuring she will display one of the highest click-through rate ads, thus maximizing her chances of generating revenue. However, this efficiency comes at the cost of aggressive bid shading by the advertisers. Anticipating the delayed allocation and the associated reduction in ad supply, advertisers strategically lower their bids, reducing revenue for the platform.
In contrast, second-price auctions rely on market thickness to generate revenue. To leverage such market thickness, the platform allocates the ad earlier, when information about quality scores is still imperfect. This leads to stronger competition, which drives up the price paid by the winning advertiser. Naturally this comes at the cost of potential inefficiency, as the platform may allocate the ad to a lower-quality advertiser due to incomplete information. Nonetheless, we show that the platform has an incentive to expedite the auction, avoiding the risk of a thin market but allocating inefficiently.

Intuitively, the divergence in equilibrium timing relies on the distinct ways each auction format leverages market thickness.
In a second-price auction, the price paid by the winner is determined by the second-highest bid. Delaying the auction increases the risk of losing bidders and thinning the market, potentially leading to a drastically lower price. The platform therefore chooses to allocate early to capitalize on the thicker market, even if it means sacrificing some efficiency. The auctioneer is somewhat \textit{averse} to information.
Instead, in a first-price auction the winner pays their own bid regardless of the competition. The auctioneer can focus on allocating efficiently by acquiring more precise quality information. The auctioneer here \textit{seeks} information. In equilibrium, this strategy incentivizes bid shading by advertisers, who anticipate facing less competition when the auction finally occurs.

The tension between allocation quality and market thickness brings a novel perspective to market design, where market thickness is often one of the top-level desiderata of the designer. The seminal work of \citet{Bulow1996} partly advocates for thicker markets as the best revenue instrument available. Instead, we show that optimal reserve prices can improve revenue substantially in this dynamic model. First, we show that, because advertisers in a \fpa no longer benefit from a thin market, the equilibrium bid-shading contracts, allowing the auctioneer to implement the optimal auction. In particular, the auctioneer no longer pays a price for the lack of commitment.  Instead, in a \spa advertisers may still benefit from a thin market, even when the auctioneer selects the optimal reserve price. The auctioneer benefits from market thickness when multiple advertisers submit large bids, thus she has an incentive to expedite the auction exercise. The advertisers, on the contrary, have an incentive to delay the exercise, thus in equilibrium they will misreport to induce delay. The \spa will never generate as much revenue as the \fpadot.

Our work contributes to the growing literature on online advertising auctions by explicitly incorporating the dynamic nature of information acquisition in conversational AI settings. We highlight the distinct ways in which first- and second-price auctions balance the trade-off between information and market thickness, offering valuable insights for platforms seeking to design revenue-maximizing ad auctions in this emerging space. At the same time, we provide one of the first studies of the \emph{endogenous} market thickness generated by a market design. We suspect that similar tradeoffs may appear in financial and asset markets, where uncertainty can be resolved at the expense of demand contraction. 

The paper is structured as follows. In \Cref{sec:model} we formally describe the model. \Cref{sec:main} presents the main results we described in this introduction: first, we compare simple auctions, first- and \spa without reserve prices. Then, we move on to ``optimal'' auctions, that is, first- and \spa with Myersonian reserve prices. We then discuss some extensions in \Cref{sec:discussion}.

\subsection{Literature Review}

We connect with a large literature in market design that studies online advertising auctions starting with \citet{edelman2007} and \citet{varian2009}. Early papers study the static problem (\citet{athey2010}, \citet{borgers2013`}), while many recent contributions analyze the repeated auction problem (\citet{balseiro2015}, \citet{chen2017}). We instead abstract away from cross-auction incentives to focus on dynamics incentives within a single auction.  The key force we study is intrinsically dynamic, and dynamics are incorporated in the auction mechanism. 

A large literature studies dynamic auctions; for a non-exhaustive list of reference see \citet{bergemann2010}. The literature mainly focuses on two sources of dynamics. The first source of dynamics is an evolving population of agents, each having fixed private information; examples include \citet{Parkes2003}, \citet{Gershkov2010} \citet{Pai2013}. The second source of dynamics is instead a fixed population of agents whose private information evolves dynamically. In our paper we have a fixed population of agents, but private information is also fixed. The only source of dynamics is in the ad quality scores, which the auctioneer exploits to optimize her revenue. 
Closer to our paper is \citet{chaves2024}, which also considers an auctioneer who chooses the timing to run the auction optimizing for market thickness. The paper considers delay as a means to increase market thickness, but such delay is costly because of time discounting. 
In our setting instead  delay \emph{reduces} market thickness, but increases the quality of the allocation. Both these effects are endogenous to the mechanism choice, and as a consequence instead of focusing on finding the optimal delay we compare equilibria of different auction formats.

There is a large literature on real options in finance and economics that deals with the revenue maximizer's optimal stopping problem; there are many reviews of this literature, such as \citet{sick1995}. In our setting, solving the optimal stopping problem is instrumental but not sufficient to characterize the auctioneer's problem, because advertisers optimally respond to the exercise timing. The equilibrium is a fixed-point problem of advertisers and auctioneer's decisions. In this sense, we connect to the literature on game-theoretic real option problems, such as those described in \citet{grenadier2000}. 
Methodologically we also adopt techniques from the continuous-time literature (see \citet{horner2017} for some excellent examples), and in particular we leverage the ``bad news'' Poisson model studied in \citet{keller2015} for some of our extensions. 

Finally, our motivation connects us to a nascent literature on Large Language Models (LLM), chatbots, and game theoretic models. As we mentioned, our problem can be thought of as the advertising decision of a LLM-based chatbot provider. Some recent papers in the literature have started thinking about the effect of LLMs on auctions and auction design  (\citet{duetting2023}, \citet{dubey_llm_2024}), and study how mechanism design should aggregate several LLM-generated input in an incentive compatible way for online advertising.
\citet{feizi2023} proposes a setup for online ads on LLMs. They mention that the system for predicting quality scores can ``update the estimate in (almost) real time, which will increase the accuracy of the
prediction''. We consider the effect of the prediction dynamics on the auctioneer's revenue. An iterative refinement of beliefs similar to the one we model here appears in \citet{harris2023}.

\section{Model}\label{sec:model}
Consider a user of a content platform with a private type $\theta = (\theta_1, \theta_2)$, where each component $\theta_i$ is drawn independently from a Bernoulli distribution with parameter $p$.
The type $\theta_i$ represents the user's interest in the ad of advertiser $i = 1,2$.\footnote{In \Cref{sec:discussion} we consider a model with more than $2$ advertisers.} Specifically, the user clicks on ad $i$ if and only if $\theta_i = 1$; in this case we say that ad $i$ is a \emph{good quality ad}. Each advertiser derives a value $v_i \sim F$ drawn independently from a regular distribution $F(v)$ over $[0,\overline{v}]$ from a successful click.

An auctioneer, who controls the platform, determines which advertiser's ad is displayed. She receives a payment from the winning advertiser if and only if the user clicks on the ad, and she can choose whether to allocate the ad according to a first- or \spadot.

The user interacts with the platform in a \emph{conversation} which unfolds over a continuous time $t \in \mathbb{R}_+$. During the conversation, the auctioneer gradually gleans insights about the quality of the ads.
We represent the conversation as a (possibly multidimensional) stochastic process $(X_t)_{t\in \mathbb{R}_+}$ with law $\nu^\theta$ which depends on the user's type $\theta$.\footnote{Formally, let  $\Omega$ be a Polish metric space and $\mathcal{F}$ a $\sigma$-algebra over $\Omega$. For each $\theta \in \Theta$ let $P^\theta$ be a probability measure over $(\Omega,\mathcal{F})$.
The collection $X = \big( X_t \colon \Omega \to \mathbb{R}^n \big)_{t \in {\mathbb{R}_+}}$ is such that every $X_t$ is measurable with respect to the Borel $\sigma$-algebra on $\mathbb{R}^n$. We denote by $\nu_t^\theta = P^\theta \circ X_t^{-1}$ the law of the random variable $X_t$ for a given parameter $\theta$.}
The auctioneer's information is given by the natural filtration induced by the conversation, denoted by $\mathbb{F} := \{\mathcal{F}_t\}_{t \in \mathbb{R}_+}$. The auctioneer then makes Bayesian inference about the value of $\theta$: her belief at time $t$ about the quality of advertiser $i$'s ad is 
\begin{equation}\label{eq:beliefs}
\mu_t^i = \mathbb{E}_t[\theta_i] = \mathbb{E}[\theta_i | \mathcal{F}_t],
\end{equation}
 and we denote by $\mu_t$ the vector $(\mu_t^1, \mu_t^2)$.\footnote{Note that each coordinate of the belief process is naturally bounded between $0$ and $1$ at all times.} Naturally, $\mu_0^i = p$, the Bernoulli prior, and $0\leq \mu^i_t \leq 1$ for all $t\geq 0$. We make few technical assumptions on the signal process $X$ to guarantee that the conversation is perfectly informative in the limit.
\begin{assumption}\label{assumption:identifiability}
Assume
\begin{enumerate}
    \item For any $t>0$ such that $\mathcal{F}_t \subset \mathcal{F}$, there exists a $t' > t$ such that  $\mathcal{F}_{t} \subset \mathcal{F}_{t'}$. 
    \item $\theta \to P^\theta(A)$ is measurable  with respect to the Borel $\sigma$-algebra on $\Theta$ for every $A \in \mathcal{F}$.
    \item There exists a measurable function $f \colon \Omega \to \Theta$ such that $f(\omega) = \theta$ almost surely with respect to $P^\theta$.
\end{enumerate}
\end{assumption}
Under these conditions, a version of Doob's consistency theorem \citep{Schwartz1965} gives the following:

\begin{proposition}\label{proposition:consistency}
If \Cref{assumption:identifiability} is satisfied, then 
$P\big( \lim_{t \to \infty} \mu_t = \theta \ |\  \theta \big) = 1$ almost surely in $\theta$.
\end{proposition}
That is, we are assuming that the platform will eventually be able to perfectly learn the type of the user.  
Note that \Cref{assumption:identifiability} is satisfied by most stochastic processes used to model news arrival. 
\begin{ex}\label{ex:bad news} 
We will later adopt a specific choice of news process, the Poisson ``bad news'' model  \citep{keller2015}: the platform independently receives no news about ad $i$ until the tick of an exponential clock with parameter $\lambda\cdot(1-\theta_i)$ for some $\lambda>0$. The arrival immediately reveals that the ad is of bad quality, i. e. $\theta_i = 0$. This two-dimensional Poisson process clearly satisfies the identifiability assumptions above. We represent a particular path of such beliefs in  \Cref{fig:poisson belief}. 
\end{ex}

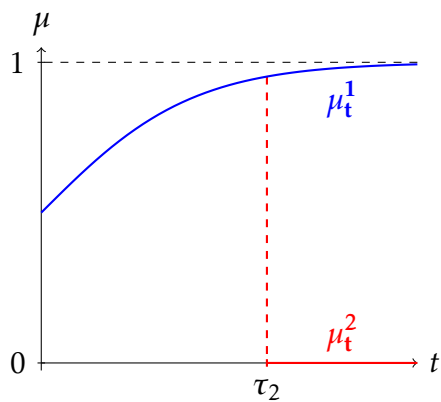
\begin{figure}[h]
    \centering
    \begin{tikzpicture}
        \draw[->] (0,0) -- (5,0) node[right] {$t$}; 
        \draw[->] (0,-0.1) -- (0,4.2) node[above] {$\mu$};
        \draw[dashed] (5,4) -- (0,4);
        \draw[dashed, red, thick] (3, {4/1.05} ) -- (3,0);
        \draw[red, thick] (3, 0.05 ) -- (3,-0.05) node[below] {\textcolor{black}{$\tau_2$}};
        \draw[red, thick] (3,0) --(5,0);
        \draw (0.05,4) -- (-0.05,4) node[left] {$1$};
        \draw (0.05,0) -- (-0.05,0) node[left] {$0$};
        
        \node[below] at (4, 3.9) {\textcolor{blue}{$\mathbf{\mu^1_t}$}};
        \node[above] at (4, 0) {\textcolor{red}{$\mathbf{\mu^2_t}$}};

        \draw[scale=1,domain=0:5,smooth,variable=\x,blue, thick] plot ({\x}, {4/(1 + exp(-\x))}); 
    \end{tikzpicture}
    \caption{This figure depicts a sample path taken by the belief process $\mu_t$ when the conversation follows the Poisson ``bad news'' model. The beliefs drift upwards in the absence of news (as is the case for the process $\mu^1_t$). Instead, news about advertiser $2$ arrived at time $\tau_2$, and thus $\mu^2_t = 0$ from $\tau_2$ onwards.}
    \label{fig:poisson belief}
\end{figure}

\subsection{The Auctioneer's Problem}

The auctioneer strategically chooses the ad to display according to either a first- or \spa, and then she decides at what time $\tau \in \Rex \coloneqq \mathbb{R}_+ \cup \{+\infty\}$ to run such auction. The timing of the allocation problem is as follows. 
\begin{itemize}
    \item First, before the conversation begins, the auctioneer announces an auction format. She commits to running either a first- or \spadot, but she cannot commit to a time $\tau$ at which such format will be run. 
    \item Then, the auctioneer solicits bids $b_1, b_2$ from the bidders (advertisers). All of the advertisers' strategic bidding decisions take place before the conversation begins.
    \item The conversation starts at time $0$, and the auctioneer begins observing the conversation $X_t$.
    \item At some time $\tau \in \Rex$, which may depends on the advertisers' bids and the realizations of the uncertainty, the auctioneer decides to ``stop'' at which point she runs the auction format she committed to.
\end{itemize}
If the winner is advertiser $i$, the auctioneer will collect revenue $\pi_{1P} = b_i\theta_i$ in a \fpa and 
$\pi_{2P} = \frac{b_{-i}\mu^{-i}_\tau}{\mu^i_\tau}\theta_{i}$ in a \spadot.\footnote{The latter expression is exactly the VCG payment for the auction where bids are weighted by the quality scores $\mu_t$.}

We make two main assumptions with respect to timing. First, advertisers cannot adjust their bids dynamically, but instead make all their strategic considerations at time $0$. Latency considerations alone justify such an assumption. Additionally, it seems unrealistic for a platform to require that advertisers monitor the conversation and update their bids based on the shape of news. Even more simply, allowing advertisers to dynamically adjust their bids based on calendar time would require a complex infrastructure on both the platform and the advertiser's end. Relaxing this assumption may be of independent interest.

Second, we assume that the auctioneer cannot commit to a time to run the auction.
We defer a discussion of this assumption in a later section.
Then, once she receives the bids of the advertisers, the auctioneer faces a real options problem. She can hold off on running the auction to learn more about the ad qualities, because she profits only when she displays a good quality ad. On the other hand, information is costly. For one, the user may leave the platform at any point in time, making the profit opportunity vanish. More subtly, delaying the auction erodes the competitive landscape. Auctions rely on competition to price the ad opportunities, but delay reduces market thickness and thus greatly diminishes competition in the auction. We show that the auction format plays an important role in the resolution of this trade-off between information and market thickness. 

\section{Auction formats, market thickness, and information}\label{sec:main}

In order to study the auctioneer's problem, we need to characterize the optimal ``exercise policy'', i.e. the optimal time in the conversation to stop and run the auction from the auctioneer's perspective, conditional on the advertisers' bids. 

\subsection{Simple auctions}\label{subsection:noreserve}
Suppose an auctioneer is facing bidders whose types come from unknown distributions. Instead of guessing the reserve price, she may simply run an auction without reserves: after all, \citet{Bulow1996} first, and more recently \citet{Hartline2009}, show that the Vickrey auction generally performs well against the optimal Bayesian auction. In this subsection we thus focus our attention on first- and \spa without reserve prices. 

Consider an auctioneer who committed to running a \spadot. She observes some reports $b_1, b_2$ and then must decide when to stop to maximize her expected revenue. Her problem can be written as
\[\max_\tau \mathbb{E}_0 \Big[ \min\big\{\mu^1_\tau b_1, \mu^2_\tau b_2\big\}\Big]\]
We are ready to prove our first result. 
\begin{lemma}\label{lemma:SPAstopping}
The optimal exercise policy for the \spa is $\tau^*_{2P} = 0$.
\end{lemma}
\begin{proof}
Following the usual Bayesian arguments, notice that the belief process $\mu^i_t$ is a martingale:  $\mathbb{E}_t[\mu_{s}] = \mu_t$ for any $s \geq t$ 
Then, by definition 
\[\mathbb{E}_t \Big[\min\big\{\mu^1_{s} b_1, \mu^2_{s} b_2\big\}\Big] \leq \mathbb{E}_t \Big[\mu^i_{s} b_i \Big] = \mu^i_t b_i\]
for all $i=1,2$ and for all $s \geq t$. 
This implies that 
\[\mathbb{E}_t \Big[\min\big\{\mu^1_{s} b_1, \mu^2_{s} b_2\big\}\Big] \leq\min\big\{\mu^1_{t} b_1, \mu^2_{t} b_2\big\}.\]
Thus, the revenue process is a super-martingale. Doob's Optional Stopping Theorem tells us that a super-martingale is optimally stopped at $0$.
\end{proof}
\Cref{lemma:SPAstopping} says that the auctioneer stops at time $0$ \emph{uniformly} over the reports of the advertisers. In particular, 1. she does not acquire any information about the user's type; and 2. she will not use time as a screening device, so the timing of the auction will not affect the incentives of the advertisers, who bid truthfully. The \spa is thus a truthful auction in this setting, and it is executed at time $0$.

The auctioneer is willing to allocate to the bidder with the highest value despite having only receiving payment from him with probability $\mu_0$. Recall that the auctioneer will profit only if the ad she displays is of good quality. By delaying the auction she could increase her chances of profiting from the displayed ad. However, the probability that any given advertiser drops out of the race ($\mu^i_t = 0$) increases with delay. The \spa relies on competition to price the ad opportunity, but as soon as the belief about one advertiser drops to $0$ competition vanishes, leaving a perfectly good opportunity to be sold for free to the remaining advertiser. The auctioneer runs the auction immediately to avoid such a catastrophic outcome. 

Intuitively, it may seem that the \spa suffers from lack of competition more than the \fpadot. That is because even if one advertiser drops out, in the \fpa the other advertiser will purchase the ad opportunity at the bid he chose at time $0$, before realizing his competition disappeared. We show next that the \fpa delays the exercise of the auction until the auctioneer is sure to allocate efficiently.
Consider then an auctioneer who committed to running a \fpadot. She observes bids $b_1, b_2$ and then decides when to stop to maximize her expected revenue, written as
\[\max_\tau \mathbb{E}_0 \Big[ \max\big\{\mu^1_\tau b_1, \mu^2_\tau b_2\big\}\Big]\]
Using this expression, we can characterize the optimal exercise policy for the \fpadot.
\begin{lemma}\label{lemma:FPAstopping}
The optimal exercise policy for the \fpa is $\tau^*_{1P} = \infty$.
\end{lemma}
\begin{proof}
By definition 
\[\mathbb{E}_t \Big[\max\big\{\mu^1_{s} b_1, \mu^2_{s} b_2\big\}\Big] \geq \mathbb{E}_t \Big[\mu^i_{s} b_i \Big] = \mu^i_t b_i\]
for all $i=1,2$ and for all $s \geq t$. 
This implies that 
\[\mathbb{E}_t \Big[\max\big\{\mu^1_{s} b_1, \mu^2_{s} b_2\big\}\Big] \geq\max\big\{\mu^1_{t} b_1, \mu^2_{t} b_2\big\}\]
for any $s \geq t$.
The revenue process of the \fpa is a sub-martingale until the first $t$ such that $\exists i \mu^i_t = 0$. Denote such $t$ by $\tilde{t}$. A sub-martingale's expected value increases over time, so the auctioneer will never stop until $\tilde{t}$. At $\tilde{t}$ the revenue process turns into the (bounded) martingale $\mu^i_t \cdot b_i$ for some $i$, and the auctioneer is indifferent between stopping and continuing at any time $t$. We say the auctioneer "stops at $\infty$", where the value of $\mu^i_\tau b_i$ at $\tau = \infty$ is given by the pointwise limit of the process $\mu^i_t b_i$ which exists almost surely because $\mu^i_t b_i$ is bounded.
\end{proof}
The \fpa is efficient: the advertiser that displays his ad is always the one with the largest $b_i \cdot \theta_i$. This auction is not truthful however: the advertisers at time $0$ foresee that, conditional of having a good quality ad, they will have no competition with probability $1-\mu_0$. Only with probability $\mu_0$ they will need to outbid their opponent, and their shading will reflect the probability of this event.
Consider the problem of bidder $i$, when bidder $-i$ is bidding according to the bidding function $\beta(v)$. For a given value $v$, bidder $i$ chooses a bid $b$ to maximize his expected profit
\[\big[\mu_0 F(\beta^{-1}(b)) + (1-\mu_0)\big](v-b).\]
Taking first-order conditions, and assuming a unique symmetric equilibrium,\footnote{We will show that there is a unique symmetric equilibrium in \Cref{proposition:fpa_continuity}.} we have the following ODE:
\[\underbrace{\beta'(v)F(v) + \beta(v)f(v) = vf(v)}_{\text{Classic ODE for FPA}} - \underbrace{\frac{1-p}{p}\beta'(v)}_{\text{Extra term}},\]
where the extra term \emph{reduces} the bid $\beta(v)$ to account for the likely absence of competition. The solution to this ODE is 
\[\beta_{FPA}(v) = \frac{1}{\frac{1-p}{p} + F(v)} \Big[ \int_0^v yf(y) dy \Big]\]
where the term $\frac{1-p}{p}$ represent the advertiser's expectations (at time $0$) about market thickness at allocation time (time $\infty$).

How does the bid shading affect the revenue to the auctioneer? The following proposition responds to this question.
\begin{theorem}\label{theorem:revenue_r0}
The revenue $\Rev_{2P}$ generated by the truthful optimally-stopped \spa dominates the revenue $\Rev_{1P}$ generated by an optimally-stopped \fpa for any value distribution $F$. In particular,
\[\frac{\Rev_{2P}}{\Rev_{1P}} = \frac{1}{\mu_0}\]
\end{theorem}
\begin{proof} 
To prove the first part of the theorem, notice that the \spa is truthful regardless of its stopping time, as long as the stopping time does not depend on the values reported by the advertisers. In particular, consider a \spa stopped at $\tau^*_{1P}$. This stopping time depends solely on the path of the beliefs $\mu_t$, but cannot be influenced by the reports $v_1, v_2$, as shown in \Cref{lemma:FPAstopping}. Then, the static Revenue Equivalence Theorem (RET) applies to auctions that run at time $\tau^*_{1P}$: the first- and \spa that stop at $\tau^*_{1P}$ allocate identically, and therefore generate the same revenue. But the auctioneer receiving truthful reports prefers to stop at time $\tau^*_{2P} < \tau^*_{1P}$, which implies that her revenue must be (weakly) larger by stopping earlier than by delaying the auction until $\tau^*_{1P}$. Therefore, the \spa generates weakly higher revenue than the \fpadot.

To prove the second part of the theorem, we leverage the virtual value characterization of the auctioneer's revenue. The expected revenue of an optimally-exercised \fpa is given by $\mu_0^2 \mathbb{E}_v \big[\max\{\phi(v_1),\phi(v_2)\}\big] + 2 \mu_0 (1-\mu_0) \mathbb{E}_v\big[\phi(v)\big]$ where $\phi(v)$ is the virtual value function. Intuitively, the \fpa allocates efficiently: with ex-ante probability $2\mu_0(1-\mu_0)$ there is only one ad which is a good match, and with probability $\mu_0^2$ the auctioneer can choose the good match with the largest virtual value.
The usual integration-by-parts argument shows that $\mathbb{E}_v[\phi(v)] = 0$, thus implying that $\Rev_{1P} = \mu_0^2 \mathbb{E}_v \big[\max\{\phi(v_1),\phi(v_2)\}\big]$. Similarly, the expected revenue of the optimally-exercised \spa is $\Rev_{2P} = \mu_0 \mathbb{E}_v \big[\max\{\phi(v_1),\phi(v_2)\} \big]$. Intuitively, the \spa allocates immediately, with probability $1$ to the bidder with the largest virtual value. She will collect her revenue only when that bidder is a good match, which happens with probability $\mu_0$.
Therefore, 
\[\frac{\Rev_{2P}}{\Rev_{1P}} = \frac{\mu_0 \mathbb{E}_v \big[\max\{\phi(v_1),\phi(v_2)\} \big]}{\mu_0^2 \mathbb{E}_v \big[\max\{\phi(v_1),\phi(v_2)\} \big]} = \frac{1}{\mu_0}\]
\end{proof}

The bid shading induced by the \fpa is detrimental to revenue, so much so that the auctioneer would rather implement a \spa and allocate inefficiently. In fact, regardless of the value distribution, the revenue obtained by a \fpa are exactly a fraction $\mu_0$ of the revenue obtained by a \spa. Thus the benefit from adopting an optimally-exercised \spa increases as the probability of encountering a good-quality ad decreases.

The case in which the user does not leave a conversation leads to a crisp, stark result. When information would otherwise be costless, the only disincentive to information acquisition is market thickness. Different auction formats leverage market thickness differently: a \spa requires competition to price the good for sale. A \fpa instead requires competition to limit the extent of the bidders' equilibrium shading. Because shading takes place ex-ante, this limits the ability of prices to react to a changing competitive landscape. The \spa retains this flexibility, so that the auctioneer can commit in a sequentially rational way not to learn anything about the competitive landscape.

\subsection{Optimal Auctions}\label{sec:reserves}

Our previous findings illustrate that the auctioneer's timing decision heavily depends on market thickness. Both first- and second-price auctions leverage market thickness for revenue generation, albeit through distinct channels. If the auctioneer knew the distribution $F$ of the bidder's values, she may carefully choose a reserve price to try and limit the impact of market thickness on her revenue, despite one of the classic insights of market design \citep{Bulow1996} highlighting how the optimal choice of reserve prices may be of second-order importance to market thickness. 
We thus study the design of first- and second-price auctions with reserve prices. With a reserve price, the auctioneer mitigates the concern of a thinning market, relying on this price floor to maintain revenue.

First, the reserve price does not affect the information acquisition strategy of the auctioneer in a \fpadot.

\begin{lemma}\label{lemma:FPAreserve}
The optimal exercise policy for the \fpa with reserve price $R$ is $\tau^*_{1P,R} = \infty$. 
\end{lemma}
\begin{proof}
The auctioneer's expected revenue when $b_1,b_2 > R$ can again be written as 
\[\max_\tau \mathbb{E}_0 \Big[ \max\big\{\mu^1_\tau b_1, \mu^2_\tau b_2\big\}\Big],\]
and we know from \Cref{lemma:FPAstopping} that the stopping time takes the form 
$\tau^*_{1P} = \infty$
Instead, if $b_i > R > b_{-i}$, the auctioneer's revenue at time $t$ is the martingale $\mathbb{E}_t[\theta_i b_i]$, so she is indifferent between stopping and continuing at all times --- in particular, we can assume she will stop at infinity with a similar argument to \Cref{lemma:FPAstopping}. 
Finally, if $b_1, b_2 < R$ the auctioneer's revenue is $0$ regardless of her stopping decision, so she is again indifferent between all stopping times, and we assume she stops at infinity.
\end{proof}

Bidders can no longer shade as aggressively because the reserve price acts as a backstop. The \fpa is no longer efficient, because of the introduction of a reserve price, but it maximizes the revenue for the auctioneer.
To see this, notice that the revenue-optimal mechanism is one that solves, for all pairs $v_1, v_2$ and user types $\theta$
\begin{equation}\label{eq:optimal auction}
    x(v_1, v_2 | \theta_1, \theta_2) \in \argmax_{x \text{ s.t. } x_1 + x_2 \leq 1} \big\{ x_1\theta_1\phi(v_1), x_2\theta_2\phi(v_2), 0\big\}.
\end{equation}
That is, the mechanism maximizes revenue if it allocates to the bidder with the highest \emph{quality-weighted} virtual value, provided it is positive. Quality-weighted virtual values can be negative only when the virtual values are negative, thus the auctioneer can implement such an allocation with an auction with a reserve price $R$ such that $\phi(R) = 0$. In particular, the allocation function of the \fpa when there exists a symmetric equilibrium in increasing bidding functions corresponds to $x(v_1,v_2)$, and thus the \fpa maximizes the auctioneer's revenue. 

The optimal reserve price has a stark effect on revenue: because the \fpa with optimal reserve is the optimal mechanism, its revenue $\pi_{1P}^R$ must be weakly larger than the revenue from a second-price auction without reserves:
$$\pi_{1P}^R \geq \pi_{2P}.$$
Together with \Cref{theorem:revenue_r0} this implies that 
\[\pi_{1P}^R \geq \frac{\pi_{1P}}{\mu_0},\]
that is, the \fpa without reserve may generate an arbitrarily small fraction of the optimal revenue.

\begin{remarking}\label{remarking:commitment}
The \fpa maximizes the auctioneer's revenue among mechanisms \emph{with and without} commitment. In other words, the optimal mechanism with commitment is implemented by a first-price auction, and moreover such a choice is sequentially rational. An auctioneer like the one we studied could have ``renegotiated'' the optimal timing of the auction at any instant $t$, deciding to stop early, but has no incentive to do so. The optimal reserve price reduces the implicit market-thickness-cost of information and does not require commitment to an optimal timing.
\end{remarking}

Perhaps surprisingly, in general the \spa instead cannot implement the optimal auction. This is because an auctioneer running a \spa cannot commit to stopping at infinity. Instead, the auctioneer sometime will have an incentive to stopping early, to capitalize on competition. To prove this, we construct an example in which the allocation induces by the optimally-stopped \spa is different from the optimal allocation function.
\begin{lemma}\label{lemma:SPAreserve}
Suppose that $X_t$ is a two-dimensional Poisson ``bad news'' model with arrival rates $\lambda(1-\theta_1)$ and $\lambda(1-\theta_2)$, with $\lambda>0$.
Let the auctioneer run a \spa with reserve price $R$. If $b_1, b_2 > 2R$, then $\tau^*_{2P,R}(b_1,b_2) = 0$. 
\end{lemma}
\begin{proof}
First, until news arrives, the belief process follows the ODE
\[\begin{cases}\dot{\mu}_t^i = \lambda \mu_t^i (1-\mu_t^i) \\
\mu_0^i = p\end{cases}\]
Because the belief is deterministic, and because we assumed symmetry, $\mu_t^1 = \mu_t^2$ until news arrives. We denote the belief prior to any arrivals by $\mu_t$.
Let without loss bidders submit bids $b_1 > b_2 > 2R$. Then the auctioneer's expected revenue from stopping at time $t$ (provided that no clock has ticked so far) is $\mu_t b_2$. Her value function at belief $\mu_t$ \emph{before any news has arrived} must satisfy 
\[V(\mu) = \max\big\{\mu_t b_2, V(\mu_{t+\Delta}) \big\}\]
for some small $\Delta > 0$. This problem can be cast as a free-boundary problem with respect to the infinitesimal generator of the belief process, which implies an interval stopping region of the form $[\overline{\mu},1]$. We will prove that $\overline{\mu} =0$. 

Suppose by contradiction that the auctioneer decided to continue at some belief $\mu>0$. Then her value function must satisfy\footnote{See \Cref{appendix} for a full derivation.}
\begin{align}\label{eq:free boundary ODE}
 V'(\mu)\mu &= 2\big(V(\mu) - \mu R\big).
\end{align}
By canonical continuous pasting arguments (see \citet{peskir2006optimal}) at belief $\overline{\mu}$ it must be that $V(\overline{\mu}) = \overline{\mu}b_2$. Smooth pasting implies instead that
\[b_2 \overline{\mu} = 2(\overline{\mu}b_2 -\overline{\mu}R )  \iff b_2= 2R.\]
When $b_2 > 2R$, there is no such belief $\overline{\mu}$. Thus, either the value function is a solution to \Cref{eq:free boundary ODE}, and it is everywhere larger than $\mu b_2$, or $V(\mu) = \mu b_2$ and the auctioneer stops immediately. 

Suppose by contradiction that the value function is a solutions to \Cref{eq:free boundary ODE}. Then, it must take the form $V(\mu) = K \mu^2 + 2R\mu$ for some $K$. To pin down $K$, note that if the belief has drifted all the way to $\mu = 1$, then the value the auctioneer can secure must be the second-highest bid, so $V(1) = b_2$.
Thus, it must be that $K = b_2 - 2R > 0$. We thus have a complete characterization of the value function, which is convex. But if this is the value function, it must pointwise weakly dominate the line $\mu b_2$, which is impossible. This is a contradiction, and it proves that the value function is $V(\mu) = \mu b_2$ when $b_2 > 2R$. The same argument can be adapted to show that $V(\mu) = (b_2 - 2R)\mu^2 + 2R\mu$ when $b_2 < 2R$, $b_2 < b_1$.

\end{proof}

The auctioneer's incentive to stop early improved her revenue with respect to a \fpa without reserve prices; instead, with reserves the \fpa dominates the \spadot.
\begin{theorem}\label{theorem:reserves}
The optimal mechanism can be implemented as a \fpa with reserves. Instead, there exist distributions $F$ such that no \spa with reserve implements the optimal mechanism.
\end{theorem}

\begin{proof}
We have already shown that the \fpa with reserve $R$ such that $\phi(R) = 0$ implements the optimal mechanism. We are left to prove that the \spa cannot possibly implement the optimal mechanism for some distribution $F$.

Take again the setting of \Cref{lemma:SPAreserve} and fix a distribution $F$ over the support $[0,\overline{v}]$ such that $\overline{v} > 2R$ where $\phi(R) = 0$, and fix  $\varepsilon < \frac{\overline{v} - 2R}{2}$.
In order for the \spa to implement the optimal auction, its allocation function must be given by \Cref{eq:optimal auction}. In particular, this implies that it must be the case that  $\tau_{2P,R} = \infty$. However, from \Cref{lemma:SPAreserve} we know that the auctioneer has an incentive to stop early if both bidders bid above $2R$. Then, if there exists a symmetric equilibrium in increasing strategies of the second-price auction that implements the optimal auction, it must be that the strategies $\beta(\cdot)$ are bounded above by $2R$.

But then, an advertiser with type $\overline{v}-\varepsilon$ has an incentive to bid $x > 2R$ instead of $\beta(\overline{v}-\varepsilon) \leq 2R$: the auctioneer will run the auction at infinity, but the advertiser will now win the item even when his opponent of type $\overline{v}$ has quality score $\theta= 1$. He will pay $\beta(\overline{v}) < 2R$ and make a profit of $\overline{v} - \varepsilon - \beta(\overline{v}) > \overline{v}- \varepsilon - 2R> \overline{v} - \frac{\overline{v} - 2R}{2} - 2R  = \frac{\overline{v} - 2R}{2}>  0$.
\end{proof}

Now that time can be used as a screening device, the \spa is no longer truthful. The equilibrium of the \spa with reserves may be complex to characterize, but it cannot possibly implement the optimal auction. Compare the result to the previous remark:  the auctioneer requires the ability to commit to timing in order to implement the same stopping rule as the first-price auction. Under that same stopping rule, the induced allocation rule would be identical to the optimal mechanism's allocation in \Cref{eq:optimal auction}. However, the auctioneer has an incentive to anticipate the auction when bids are large relative to the reserve price, to capitalize on the current market thickness.

\section{Discussion}\label{sec:discussion}

\paragraph{Number of advertisers.} We purposely chose to study a model with two advertisers because it allows us to write quite the general results. However, it should be clear that the argument laid out in \Cref{lemma:SPAstopping} cannot be immediately replicated with more than two advertisers. 
In this section we trade off some degree of generality by specializing to the exponential ``bad-news'' model, but we relax the number of advertisers competing for the ad slot to $n>2$. The user then has type $\theta = (\theta_1, \dots, \theta_n)$, and the news process is an $n$-dimensional Poisson process with arrival rates $\lambda(1-\theta_i)$ on the $i$-th coordinate. Let the first arrival time of the Poisson process associated with advertiser $i$ be denoted by $\tau_i$.
We can show the following:
\begin{lemma}\label{lemma:multiple_advertisers}
For any realization of the exponential clocks $\tau_1, \dots, \tau_n$ the \spa optimally stops earlier than the \fpadot, i.e. 
\[\tau^*_{2P} \leq \tau^*_{1P} \quad \text{ pointwise.}\]
\end{lemma}
\begin{proof}
To see this, first note that the revenue from the \fpa is a maximum of $n$ martingales, and thus is a sub-martingale. Let $K_t$ be the number of advertisers $i$ such that $\mu_t^i \neq 0$. 
Then, 
\[\tau_{1P}^* = \inf_t \{K_t |  K_t = 1\}.\]
That is, the auctioneer stops in the limit when there are at least two advertisers who are a good match. She stops earlier when there are less than two good matches and only one competitor remains in the race.\footnote{This is equivalent to a stopping time $\tau_{1P}^* = +\infty$, because when $K_t = 1$ the revenue process is a martingale.} 

Instead, the revenue from a \spa is no longer a super-martingale. There is a natural upper bound on the stopping time of the \spa, given by our previous results: suppose that the conversation has produced bad news for all but two advertisers. Then, the auctioneer finds herself in the same position as she was in before, and she will stop. Therefore, 
\[\tau^*_{2P} \leq \inf_t{K_t | K_t = 2}.\]
One can show that the inequality can be strict by constructing the Hamilton-Jacobi-Bellman equation for the stopping problem when $K_t = 3$. Up to relabeling, we have $b_1 \geq b_2 \geq b_3$, and so the value function in the continuation region is characterized by:
\[3V(\mu) = \dot{V}(\mu)\mu + \mu(b_2 + 2b_3).\]
The threshold stopping belief is $\overline{\mu} = 0$. If $b_2< 2b_3$, the value function is then $V(\mu) = \frac{b_2 - 2b_3}{2}\mu^3 + \frac{b_2 + 2b_3}{2}\mu$ and the auctioneer continues until bad news arrives (if ever). Otherwise, if $b_2 \geq 2b_3$, the value function is $V(\mu) = \mu b_2$ and the auctioneer stops immediately. Intuitively, if the lowest bid $b_3$ is sufficiently worse than the second-highest ($b_2$), continuing is suboptimal because the problem is sufficiently similar to a two-bidder auction. 
\end{proof}
The intuition for this result is similar to \Cref{lemma:SPAreserve}. There, instead of an additional bidder with his own quality $\theta_3$ we added a ``fake'' bidder, with bid $R$ and known quality equal to $1$. It seems only natural that if such a bidder was not sufficient to induce the auctioneer to stop at $t = \infty$, then a weakly worse bidder (such as one with an uncertain quality) will not suffice either.
From \Cref{lemma:multiple_advertisers} and the Revenue Equivalence Theorem (which again holds because of the exponential bad news assumption) immediately follows that the revenue generated by the \spa is larger than the revenue generated by a \fpadot. 

Our results hold also when the auctioneer has a prior about the advertisers' values. The \fpa with optimal reserve dominates the \spa with optimal reserve. To see this, we can show that the \fpa implements the optimal auction in this domain, and instead the \spa does not. The proof is a replica of the the one discussed in the previous section and is thus omitted.

\paragraph{Explicit cost of information.}
A natural question is how introducing an explicit cost of information acquisition would affect the results. Small information acquisition costs are particularly natural because they can also be interpreted as the instantaneous rate of departure of the conversational AI user. Essentially, extending the model in this direction relaxes the assumption that the conversation will eventually perfectly learn the ad qualities, and adds an incentive to avoid delay in the auctioneer's problem.
The analysis is instructive and we include it in \Cref{subsection:discount}. We again specialize to the exponential news model, and we find that many of our results are robust to time discounting.

\section{Conclusion}

In this paper, we investigate the effect of endogenous market thickness and how it influences optimal auction design. Two key countervailing incentives for the auctioneer determine when the auction should be held:
\begin{itemize}
\item Information Acquisition: Over time, signals about ad quality improve, allowing the auctioneer to make better allocation decisions.
\item Market Thickness: As quality signals improve, advertisers separate between more and less competitive. As the market thins, revenue deteriorates.
\end{itemize}
We prove that the \spa is most sensitive to market thickness:  the auction happens early when quality scores are very uncertain. In the \fpa instead, the auctioneer is most sensitive to information.  She wants to learn as much as possible before running the auction, but the bidders anticipate the auctioneer's behavior and shade their bids aggressively. The auctioneer generates more revenue by running a \spadot.
Interestingly, the introduction of reserve prices wipes the \spadot's advantage. The \fpa implements the optimal auction, while the \spa is not truthful and sometimes ends too early. 

We are motivated by the monetization problem of Generative AI providers, and the model tightly reflects our fundamental question. However, it turns out that the trade-off between competition and quality scores appears in a number of interesting economic settings. We propose two examples below of auctions with \emph{dynamic} scoring, where the winner selection procedure relies on factors other than the buyers' values, and such factors are dynamically updated.

\paragraph{Mergers \& Acquisitions.} Multiple firms are interested in purchasing a small business. In most competitive M\&A processes, buyers submit bids to acquire the business with some contingencies that will be resolved during the due diligence period.\footnote{ See \citet{marquardt2015role} and \citet{wangerin2019m}.} Such contingencies may result in contract termination.
While due diligence is costly, a natural question is whether bidders should be encouraged to perform due diligence \emph{before or after} the winner selection. Selecting a winner before due diligence makes the process more competitive, but exposes the business to the risk of a failed M\&A deal. Viceversa, allowing due diligence to take place before selecting the winner may thin out the pool of interested parties, thereby reducing competitive pressure.

\paragraph{Public Procurement.} Government agencies procure goods and services from suppliers in an auction. Regulators require that suppliers pass certain probity measures -- for example, money-laundering and criminal record checks.
Should the agencies front-load or back-load compliance verification when procuring from industry suppliers?
That is, should the agencies run their auction and then audit the winner, or should they audit all participants and only after choose the winner? 
EU directives only require a bidder self-declaration at auction time, and then demand full compliance checks to be performed on the winner.\footnote{Directive 2014/24/EU, art. 57-61.} Again, by performing compliance checks after the auction, the agencies take advantage of market thickness to procure at low prices, at the risk of unfulfilled contract. Instead, running the auction after thorough checks would ensure contract fulfillment but would possibly increase prices. 
\vspace{1em}

Our theory has interesting implications for ad auction design. In light of the rapid adoption of interactive systems where platforms can dynamically learn the preferences of users, platforms that are planning to monetize using ad-auctions should think carefully about the dynamic forces we highlight when selecting an auction format.

% Bibliography
\bibliographystyle{ACM-Reference-Format}
\bibliography{references}

% Appendix
\appendix

\section{Extensions}\label{section:extensions}
In this appendix we explore two natural extensions to the baseline model. First, we show that our main insights carry over to a model where information acquisition has a direct cost, in addition to the indirect, market-thickness, cost. 

\subsection{Costly Information Acquisition ($r>0$)}\label{subsection:discount}
The problem of the auctioneer becomes more complex when she faces two separate costs of information. As before, she pays for information indirectly, by reducing market thickness through delay. Additionally, now information is costly in and of itself, because the user may leave the platform at a rate $r>0$.

This additional force seems to push the auctioneer toward less delay. However, the same discounting applies to the advertisers, who may be tempted to shade their bids less. Receiving the allocation in the future is less valuable in this world, so they may want to try and influence the timing of allocation by reporting different bids. We begin by showing that such incentives are not sufficient to change the auctioneer's decision in the optimally-stopped \spadot.
\begin{proposition}\label{proposition:SPAstopped_r}
The optimal exercise policy for the \spa is $\tau^*_{2P} = 0$.
\end{proposition}
\begin{proof}
For any pair of reports $v_1, v_2$ the auctioneer faces the following optimization problem:
\[\max_{\tau} \mathbb{E}_0\Big[ e^{-r\tau} \min\{\mu^1_\tau b_1, \mu^2_\tau b_2\}\Big]\]
We know from \Cref{lemma:SPAstopping} that $\min\{\mu^1_\tau b_1, \mu^2_\tau b_2\}$ is a super-martingale, and the auctioneer here is trying to maximize a discounted super-martingale, which is itself a super-martingale. The result then follows exactly as in \Cref{lemma:SPAstopping}.
\end{proof}

Because this result again holds uniformly over reports, the \spa simply allocates at time $0$ to the highest bidder. This result is independent of the discount factor, which is particularly striking: the auctioneer allocates immediately, and is unwilling to take any risk with respect to the realization of the quality process. In a world in which competition may suddenly collapse (as it happens when one of the clocks ticks) the \spa hedges against such risk by sacrificing efficiency. With probability $1-\mu_0$ the auctioneer will not generate any revenue, but with probability $\mu_0$ she will do so at a time in which competition is strong. Finally, note that because the \spa allocates at a time independent of the reports of the advertisers, advertisers have no incentive to misreport, and the auction is truthful. 

Next, we describe the optimal exercise policy of the \fpa as a function of the advertisers' bids.
First note that the auctioneer will never stop after $\min\{ \tau_1, \tau_2 \}$. In particular, suppose the auctioneer has not stopped yet at $\tau_1 < +\infty$, when the clock of advertiser $1$ ticks. Then the auctioneer will stop exactly at $\tau_1$, because the auctioneer's expected revenue at time $t>\tau_1$ is given by $\mathbb{E}_{\tau_1} \big[ e^{-rt}\mu_t^2 b_2\big]$ (which is a super-martingale) since advertiser $1$ is out of the race. 
Thus, we only need to consider the case in which no news has arrived yet. In this case, the state space is completely characterized by a single number $\mu$, because $\mu_t^1 = \mu_t^2$ before the arrival of news. 
The auctioneer's value at beliefs $\mu^1_t = \mu^2_t = \mu$ solves the Hamilton-Jacobi-Bellman equation below. Its derivation is in \Cref{appendix}.

\begin{align*}
    \underbrace{V'(\mu)\lambda\mu(1-\mu)}_{\text{gain from delay}} = \overbrace{r V(\mu)}^{\text{direct cost of delay}} + \underbrace{(1-\mu)\lambda}_{\text{clock ticks}}\Big(\overbrace{V(\mu) - \mu b_1}^{\text{cost if ad $2$ is bad}} + \overbrace{V(\mu) - \mu b_2}^{\text{cost if ad $1$ is bad}}\Big)
\end{align*}

The terms in the HJB highlight the balance between benefits and costs of information. The left-hand side represents the increase in continuation value induced by a marginal delay ---since beliefs are drifting upwards, if the clock deosn't tick, the expected revenue increases. The right-hand side represents the decrease in continuation value over the marginal delay --- first because of information cost (in terms of discounting) and second because a tick of the clock shatters competition.
To find out when the auctioneer will run the auction, without loss we set $b_1 \geq b_2$. By requiring continuous pasting on the stopping boundary we find that, if no clock has ticked yet, the auctioneer will stop and run the auction when the beliefs $\mu^1_t = \mu^2_t$ reach the threshold 
$\overline{\mu} = \max\big\{1-\frac{r}{\lambda}\frac{b_1}{b_2}, \mu_0 \big\}$.
That is, unless a clock ticked the auctioneer runs the auction at 
\[\tau^{no-news}(b_1,b_2) = \begin{cases}
 \frac{1}{\lambda}\log\Big(\frac{\lambda}{r}\frac{\min_i b_i}{\max_i b_i} - 1\Big)  \quad \text{ if } 1-\frac{r}{\lambda}\frac{\max_i b_i}{\min_i b_i} > \mu_0\\
0 \hspace{8.7em} \text{otherwise.}
\end{cases}\]
and allocates to the highest bidder. From this equation we glean some insight into the trade-offs faced by the auctioneer. First, the auctioneer runs the auction with the largest delay when both bidders' bids coincide. Intuitively, when the bids are close to each other, the auctioneer has little to lose from one of the advertisers dropping out of the race. Instead, she wants to anticipate the auction when one bid is substantially larger than the other, because her revenue is now more at risk. This can be seen by noting that the threshold $\overline{\mu}$ is decreasing in the ratio $\frac{b_1}{b_2}$.

This concludes the characterization of the optimal stopping time for the \fpa, which we report below.
\begin{proposition}\label{proposition:FPAstopped_r}
The optimal exercise policy for the \fpa is 
\[ \tau^*_{1P}(b_1,b_2) = 
\min \Big\{ \tau^{no-news}(b_1,b_2),\tau_1,\tau_2 \Big\}
\]
\end{proposition}

Unfortunately, computing the equilibrium of the \fpa requires aggregating exercise policies over the set of possible reported pairs of bids. This is a largely intractable problem, even for simple value distributions.
We prove however that when the probability of the user leaving is sufficiently small, the revenue from a \spa dominates the revenue from a \fpa. In other words, when the advertisers and the auctioneer become more patient (or the beliefs converge sufficiently fast to the truth) the \fpa allocates almost efficiently, and thus its revenue cannot possibly be close to the revenue of the inefficient \spa.

\begin{proposition}\label{proposition:fpa_continuity}
For every discount rate $r \geq 0$, there exists a pure strategy equilibrium for the \fpa game. Moreover, for $r=0$ the unique equilibrium is symmetric.
\end{proposition}

\begin{proof}
Observe that the our dynamic \fpa, can be summarized as pay your bid auction where the payoff for bidder $i$ with valuation $v_i$ and given bids $b_1, b_2$ is $u_i(v_i,b_i,b_{-i} = \mu_0 \cdot (v_i-b_i) x_i(b_i,b_{-i}|r)$ where
\begin{align*} x_i(b_i,b_{-i}|r) =  \mathbb{E}_0 \big[ & \mathbf{1}_{\{\tau_{-i}\geq \tau^{no-news}(b_i,b_{-i}|r)\}} e^{-r \tau^{no-news}(b_i,b_{-i}|r)} \big(\mathbf{1}_{\{b_{i}>b_{-i}\}}  +\frac 1 2\cdot + \mathbf{1}_{\{b_i=b_{-i}\}}\big) \\
& +  \mathbf{1}_{\{\tau_{-i}\geq \tau^{no-news}(b_{i},b_{-i}|r)\}} e^{-r\tau_{-i}} \big].
\end{align*}

Observe that for bids $b_i\neq b_{-i}$ the payoff is continuous on $b_i$ as the clock $\tau_1,\tau_2$ are independent of the bids and the optimal policy $\tau^{no-news}(b_i,b_{-i}|r)$ is continuous on $b_i$. Therefore, the same proof of the multi-unit pay your bid auction in Example 5.2 in \citet{reny99} applies to our setting which guarantees the existence of a pure strategy equilibrium. 

For $r=0$, notice that $\tau^{no-news}(b_1,b_2|r)=\infty$, and therefore the allocation rule of bidder $i$ only depends on the ranking of the bids but not on the value of the bids. Thus, using the language of \citet{chawla_hartline_2013}, the auction is a rank-based allocation rule. Moreover, the first price nature of the auction implies that it satisfies the bid-based payment and win-vs-tie-strict properties required for Theorems 3.1 and 4.6 in \citet{chawla_hartline_2013} which imply that there is a unique equilibrium for the game and bidders use the same pure strategy bidding function.
\end{proof}

The key idea of the existence proof relies on the continuity of $\tau^*_{1P}(b_1,b_2)$ on the bids for a fixed discount rate $r$. In some sense the dynamic \fpa then inherits the same continuity properties of the static \fpa, which is sufficiently continuous to apply the machinery of \citet{reny99}.\footnote{See for instance the paper's Example 5.2, that shows that for static first-price auctions there exists a pure strategy equilibrium with strictly increasing bidding strategies.}

To obtain equilibrium uniqueness for the case when $r=0$, we leveraged Theorems 3.1 and 4.6 in \citet{chawla_hartline_2013}. However, a key assumption for their results is that the allocation depends only on the ranking of the bids. When $r>0$ instead the auction timing is bid-dependent, and thus from an ex-ante perspective bids affect the allocation rule. This is not the case when $r=0$, allowing us to conclude uniqueness and symmetry of equilibrium in that case.

To show that that the equilibrium is continuous as function of the discount rate, we consider the mixed extension of the auction game.\footnote{We refer to Chapter 2 of \citep{parthasarathy2005probability} for textbook treatment on the space of Borel measures and the weak topology on it.}  We denote the Borel measures on $[0,1]^2$ by $\mathcal{B}([0,1]^2)$ so that a strategy $B\in \mathcal{B}([0,1]^2)$ corresponds to a probability measure $(v,b)\sim B$ with $\Pr_B[v'\leq v] = F(v)$ for all $v\in [0,\overline v]$ (almost surely). For example, a pure strategy $b(v)$ has mixed strategy representation $B_{b}$ where $\Pr_{B_{b}}[(v,b): b= b(v)] = 1$. 

We endow $\mathcal{B}([0,1]^2)$ with the weak topology, so that $B_n \Rightarrow B$ if for all continuous function $f:[0,1]^2\to R$ we have that $\int f dB_n \to \int f dB$. Because $\mathcal{B}([0,1]^2)$ is compact and the subsspace of mixed strategies is closed, we conclude that the subspace of mixed bidding strategies is compact.

\begin{claim}\label{claim1}
Given a discount rate $r\geq 0$, consider an equilibrium $(B^*_1(r),B^*_2(r))$. Then the probability of having a tie on the equilibrium is zero. 
\end{claim}
\begin{proof}
Suppose that by contradiction that there are ties on the equilibrium. So consider a type $v_i$ on where there is a tie. By equilibrium optimality we must have that $\Pr_{B^*_i(r)}[(v_i,b): b\geq v_i]=0$: else, the bidder obtains a payoff of zero while by bidding small $\epsilon$ obtains a positive payoff since it wins the object with positive probability whenever the other bidder has a valuation $v_{-i}\leq \epsilon$. Since type $v_i$ is bidding less that its value, if there is a tie it can raise its bids by a small $\epsilon$. By doing that it removes the ties and gets a mass increase on their payoff, and because $\tau^*_{1P}$ is continuous the time allocation changes smoothly so that would be a profitable deviation. Therefore, there are no ties on the equilibrium.
\end{proof}

\begin{claim}\label{claim2}
Consider $r_n \to 0$ and equilibrium bids $(B^*_1(r_n),B^*_2(r_n))$ then $(B^*_1(r_n),B^*_2(r_n)) \Rightarrow (B_{b^*},B_{b^*})$, where $b^*(\cdot)$ is the symmetric bidding strategy in the \fpa game when $r=0$. Furthermore, if $B^*_i(r_n)$ is a a pure strategy for all $n$, then $b^*(v|r_n)\to b^*(v|0)$ for all $v\in [0,\overline v]$ (a.s.).
\end{claim}

\begin{proof}
Because the space of actions is compact, consider a subsequence of $(B^*_1(r_{n_k}),B^*_2(r_{n_k})) \Rightarrow (B^*_1(0),B^*_2(0))$. We assert that $(B_1(0),B_2(0))$ is a nash equilibrium in the mixed extension of the game, which by \Cref{proposition:fpa_continuity} implies that $B_i(0) = B_{b^*}$. This further implies that is the unique accumulation point of the sequence $(B^*_1(r_n),B^*_2(r_n))$ which, given that the space is compact, implies that $B^*_i(r_n) \Rightarrow B_{b^*}$.

Next, observe that $e^{-r \tau^{no-news}(b_i,b_{-i}|r_n)}$ is increasing as $n\to \infty$ and converges pointwise to 1, thus by Dini's Theorem we have that it converges uniformly to $1$. The same convergence result can be used for the indicator functions  $\mathbf{1}_{\{\tau_{-i}\geq \tau^{no-news}(b_{i},b_{-i}|r)\}}$, $ \mathbf{1}_{\{\tau_{-i}\geq \tau^{no-news}(b_{i},b_{-i}|r)\}}$ which implies that $x_i(b_i,b_{-i}|r_n)$ converges uniformly to  $x_i(b_i,b_{-i}|r_n)$ and therefore that $u_i(v_i,b_i,b_{-i}|r_{n_k})$ converges uniformly to $u_i(v_i,b_i,b_{-i}|r_{n_k})$.

Since there are no ties on the equilibrium (\Cref{claim1}) so that the payoff functions are continuous when there are no ties (see proof of \Cref{proposition:fpa_continuity}), and the payoff function has uniform convergence, we have that $U_i(B^*_1(r_{n_k}),B^*_2(r_{n_k})|r_{n_k}) \to U_i(B_1(0),B_2(0)|r=0)$ where
$$ U_i(B_1(r),B_2(r),r) = \int \int u_i(v_i,b_i,b_{-i}|r) dB_{-i}(v_{-i}, b_{-i})dB_i(v_i, b_i).$$

To show that $(B_1(0),B_2(0))$ is a Nash Equilibrium, consider a pure strategy $b$. First, suppose that $b$ is increasing on $v$. Observe that
$U_i(B^*_i(r_{n_k}),B^*_{-i}(r_{n_k})|r_{n_k}) \geq U_i(B_{b},B^*_{-i}(r_{n_k})|r_{n_k}) $ since $(B^*_i(r_{n_k}),B^*_{-i}(r_{n_k}))$ is a Nash Equilibrium when the discount rate is $r_{n_k}$. Because $b$ is increasing we know that there are no ties so the payoff function is continuous and has uniform converge as $n_{k}\to\infty$. Therefore, by taking the limit we conclude that $U_i(B_1(0),B_2(0)|r=0) \geq U_i(B_{b},B_2(0)|r=0) $. If the function $b$ is non-decreasing, then for every $\epsilon>0$ an increasing function $\tilde b\geq b$  exists such $U_i(B_{b},B_2(0)|r=0) \leq  U_i(B_{\tilde b},B_2(0)|r=0) + \epsilon$. We conclude by taking $\epsilon \to 0$, that for all feasible bidding strategies $U_i(B_1(0),B_2(0)|r=0) \geq U_i(B_{b},B_2(0)|r=0) $. Therefore, the uniqueness of equilibrium when $r=0$ implies that $B_{i}(0)= B_{b^*}$.
 
 Finally, if $B^*_i(r_n)$ is a sequence of pure strategy equilibria. Consider $f_\epsilon $ a continuous approximation of the indicator function $\mathbf{1}_{[v-\epsilon,v+\epsilon]}$. Because $B^*_i(r_n) \Rightarrow B_{b^*}$ we have that $\int f_\epsilon dB^*_i(r_n) \to \int f_\epsilon dB_{b^*}$. Thus, by taking $\epsilon$ we get the desired pointwise convergence result.
\end{proof}

With the existence of equilibrium at hand and thanks to the fact that the optimal stopping time $\tau^*_{1P}$ is also continuous on the discount rate, we have shown that the equilibrium bidding strategies are continuous on the discount rate. Therefore we can show a generalized version of \Cref{theorem:revenue_r0} for the case where the information cost -- measured by the discount rate -- is small. 

\begin{theorem}\label{theo:revenue_fpa_spa_positive_discount_rate}
Fix a value distribution $F$, there exists a discount rate $\underline{r}>0$ such that for any $r < \underline{r}$ the revenue generated by the \spa dominates the revenue generated by any equilibrium of the \fpadot.
\end{theorem}

\begin{proof}
From \Cref{proposition:SPAstopped_r}, we have that the optimal stopping is independent of the discount rate $r=0$, which in turn implies that the auctioneer's revenue when optimally running a \spa is independent of $r$. Thus, $\pi^*_{2P}(r) = \pi^*_{2P}(0)$. 

For the \fpa case, we want to claim that for $r_n\to0$, and any mixed equilibrium of the game when $r_n>0$ we have that $\pi^*_{1P}(r_n) \to \pi^*_{1P}(0) $. 
Indeed, consider a sequence of mixed bidding equilibrium $(B^*_1(r_n),B^*_2(r_n))$. From \Cref{claim2}, we have that $B^*_1(r_n)\Rightarrow B_{b^*}$, where $b^*$ correspond to the symmetric bidding strategy in the unique equilibrium of the \fpa game when $r=0$ (\Cref{proposition:fpa_continuity}). Fixing $\tau_1,\tau_2$,  we have that $\tau^*_{1P}(b1,b2|r_n)\to \min\{\tau_1,\tau_2\}:= \tau_1\wedge \tau_2$ converges pointwise and $\tau^*_{1P}(b1,b2|r_n)$ is increasing on $r_n$. Then, Dini's Theorem implies that the convergence is uniformly on $b_1,b_2$.  Therefore, the following limit holds
\begin{align*}
 &\lim_{n\to \infty}    \int \int  e^{-r \tau^*_{1P}(b_1,b_2|r_n)} \min\{\mu^1_\tau b_1, \mu^2_\tau b_2\}dB^*_1(r_n) dB^*_2(r_n) \\
 & = \int \int \min\{\mu^1_{\tau_1\wedge \tau_2} b_1, \mu^2_{\tau_1\wedge \tau_2} b_2)\} dB^*_{b^*} dB^*_{b^*}\\
& =  \min\{\mu^1_{\tau_1\wedge \tau_2} b^*(v_1), \mu^2_{\tau_1\wedge \tau_2} b^*(v_2)\}
\end{align*}
By taking expectations on $\tau_1,\tau_2$, we conclude that $\pi^*_{1P}(r_n) \to \pi^*_{1P}(0) $.

Finally, we invoke \Cref{theorem:revenue_r0} that shows that $\pi^*_{1P}(0)<\pi^*_{2P}(0)$ and conclude that a $\underline r>0$ exists such that for $r\leq \underline r$, $\pi^*_{1P}(r)<\pi^*_{2P}(r)$.
\end{proof}

\section{Calculations for Poisson Model}\label{appendix}

The evolution of beliefs in a Poisson model follow the derivation below.
\begin{align*}
    \mu_{t+\Delta}^i &= \mathbb{P}(\theta_i = 1 | \text{ no news before } t+\Delta) \\
    &= \frac{\mathbb{P}(\text{no news in } [t,t+\Delta)) \mathbb{P}(\theta_i = 1 | \text{ no news before } t)}{\mathbb{P}(\text{no news in } [t,t+\Delta | \text{ no news before } t)} \\
    &= \frac{\mu^i_t}{\mu^i_t + (1-\mu^i_t)e^{-\lambda \Delta}}
\end{align*}
Then 
\begin{align*}
    \dot{\mu}^i_t = \lim_{\Delta \to 0} \frac{\mu_{t+\Delta}^i - \mu^i_t}{\Delta} &= \lim_{\Delta \to 0}\frac{\mu^i_t}{\Delta} \bigg(\frac{1}{\mu^i_t + (1-\mu^i_t)e^{-\lambda \Delta}} -1 \bigg) \\
    &= \lim_{\Delta \to 0} \frac{\mu^i_t}{\Delta} \bigg(\frac{1 - \mu^i_t - (1-\mu^i_t)e^{-\lambda \Delta}}{\mu^i_t + (1-\mu^i_t)e^{-\lambda \Delta}}\bigg) \\
    &= \lambda \mu^i_t (1-\mu^i_t)
\end{align*}

Fix $b_1 \geq b_2$. In a \fpa the value function for non-zero beliefs is given by
\begin{align*}
    V(\mu_t, \mu_t) = \max &\bigg\{ b_1 \mu_t,     e^{-r\Delta}V(\mu_{t+\Delta}, \mu_{t+\Delta})(\mu_t^2 +2\mu_t(1-\mu_t)e^{-\lambda \Delta} +e^{-2\lambda \Delta}(1-\mu_t)^2) \\
    &+ e^{-r\Delta}\Big(\mu_t(1-\mu_t)(1-e^{-\lambda\Delta})(b_1 \mu_t + b_2 \mu_t) + (1-\mu_t)^2(1-e^{-2\lambda\Delta})\big(\frac{b_1 \mu_t + b_2 \mu_t}{2} \big) \Big)\bigg\}
\end{align*}
while in a \spa it is given by 
\begin{align*}
    V(\mu_t, \mu_t) = \max &\bigg\{ b_2 \mu_t,     e^{-r\Delta}V(\mu_{t+\Delta}, \mu_{t+\Delta})(\mu_t^2 +2\mu_t(1-\mu_t)e^{-\lambda \Delta} +e^{-2\lambda \Delta}(1-\mu_t)^2)\bigg\}.
\end{align*}
For a \spa with reserve price $R<b_2$, the value function is given by 
\begin{align*}
    V(\mu_t, \mu_t) = \max &\bigg\{ b_2 \mu_t,     e^{-r\Delta}V(\mu_{t+\Delta}, \mu_{t+\Delta})(\mu_t^2 +2\mu_t(1-\mu_t)e^{-\lambda \Delta} +e^{-2\lambda \Delta}(1-\mu_t)^2) \\
    &+ e^{-r\Delta}R \mu_t\Big(2\mu_t(1-\mu_t)(1-e^{-\lambda\Delta})+ (1-\mu_t)^2(1-e^{-2\lambda\Delta})  \Big)\bigg\}
\end{align*}
From here onward, we proceed with calculations for the \fpa only. At a belief where the auctioneer wants to continue,
\begin{align*}
    V(\mu_t) &=    e^{-r\Delta}V(\mu_{t+\Delta})\big(\mu_t + (1-\mu_t)e^{-\lambda \Delta}\big)^2 \\
    &+ e^{-r\Delta}(b_1 \mu_t + b_2 \mu_t) \bigg( \mu_t(1-\mu_t)(1-e^{-\lambda\Delta}) +  \frac{(1-\mu_t)^2}{2} (1-e^{-2\lambda\Delta})\bigg)
\end{align*}
\begin{align*}
    0 &= \frac{e^{-r\Delta}V(\mu_{t+\Delta})\big(\mu_t + (1-\mu_t)e^{-\lambda \Delta}\big)^2  - V(\mu_t)}{\Delta} \\
    &+ \frac{e^{-r\Delta}}{\Delta}(b_1 \mu_t + b_2 \mu_t) \bigg( \mu_t(1-\mu_t)(1-e^{-\lambda\Delta}) +  \frac{(1-\mu_t)^2}{2} (1-e^{-2\lambda\Delta})\bigg)
\end{align*}
Taking the limit for $\Delta \to 0$, and writing $\rho = \frac{r}{\lambda}$ we get
\begin{align*}
    V(\mu_t)\big(r + 2\lambda(1-\mu_t)\big)  &=  V'(\mu_t)\dot{\mu}_t + (b_1+b_2)\mu_t\Big(\mu_t(1-\mu_t)\lambda +(1-\mu_t)^2\lambda \Big) \iff \\
\iff    V'(\mu)\mu(1-\mu) &= \rho V(\mu) + (1-\mu)\Big(2V(\mu) - \mu(b_1 + b_2)\Big)
\end{align*}
Note that if $r = 0$ the ODE becomes
\[V'(\mu) \mu = 2V(\mu) - \mu(b_1 + b_2),\]
similar to what we used in \Cref{lemma:SPAreserve}.
The solution to the general ODE is
\[V(\mu) = \frac{\mu(1-\mu)}{\rho + 1}(b_1 + b_2) + K\mu^2\bigg(\frac{\mu}{1-\mu}\bigg)^\rho\]
where 
\[K = \frac{b_1 }{(1+\rho)}\bigg(\frac{b_2}{\rho b_1} - 1\bigg)^{-\rho}\]
which results in 
\[V(\mu) = \frac{\mu(1-\mu)}{\rho + 1}(b_1 + b_2) + \frac{b_1 }{(1+\rho)}\bigg(\frac{b_2 - \rho b_1}{\rho b_1}\bigg)^{-\rho}\mu^2\bigg(\frac{\mu}{1-\mu}\bigg)^\rho\]
We find this value of $K$ by requiring continuous pasting at the threshold $\bar{\mu}$ such that $\bar{\mu}b_1 = V(\bar{\mu})$. This threshold turns out to be such that $1-\bar{\mu} = \frac{\rho b_1}{b_2}$. Then we can rewrite 
\[K = \frac{b_1 }{(1+\rho)}\bigg(\frac{1}{1-\bar{\mu}} - 1\bigg)^{-\rho} = \frac{b_1 }{(1+\rho)}\bigg(\frac{1-\bar{\mu}}{\bar{\mu}}\bigg)^{\rho}\]
and 
\[V(\mu) = \frac{\mu(1-\mu)}{\rho + 1}(b_1 + b_2) + \frac{b_1 }{(1+\rho)}\mu^2\bigg(\frac{\mu(1-\bar{\mu})}{\bar{\mu}(1-\mu)}\bigg)^\rho\]

\end{document}